\documentclass{amsart}
\usepackage{amsmath, amsthm, mathrsfs}
\usepackage{graphicx}
\usepackage{hyperref}
\usepackage{enumitem}
\usepackage{mathabx}
\usepackage{txfonts}
\usepackage{array}
\usepackage{tikz}
\usetikzlibrary{arrows.meta, decorations.markings, calc}
\usepackage[linesnumbered,ruled,vlined]{algorithm2e}

\newcommand{\dashint}{%
  \mathop{%
    \int
  }}

\newcounter{thm}
\newcounter{ex}
\newcounter{re}
\newtheorem{theorem}[thm]{Theorem}
\newtheorem{lemma}[thm]{Lemma}
\newtheorem{corollary}[thm]{Corollary}
\newtheorem{proposition}[thm]{Proposition}

\newtheorem{remark}[thm]{Remark}
\newtheorem{definition}[thm]{Definition}

\begin{document}

\title[$W^{s,p}$ loop spaces]{Poisson structures on weak Sobolev loop spaces and applications to integrable systems}

\author{Jean-Pierre Magnot}

\address{{SFR MATHSTIC, LAREMA, Universit\'e d’Angers, 2 Bd Lavoisier, 
49045 Angers cedex 1, France;  Lyc\'ee Jeanne d'Arc, 40 avenue de Grande Bretagne, 63000 Clermont-Ferrand, 
France}; Lepage Research Institute, 17 novembra 1, 081 16 Presov, Slovakia}
\email{\small magnot@math.cnrs.fr; jean-pierr.magnot@ac-clermont.fr}

\begin{abstract}
We develop a framework for Poisson geometry on loop spaces of low regularity,
extending Mokhov’s classical constructions from smooth loops to weak Sobolev
spaces $W^{s,p}(\mathbb S^1;\mathbb R^m)$ with $0<s\le \tfrac12$ and $1<p<\infty$.
Within this setting we construct presymplectic and Poisson structures of
hydrodynamic type, as well as their weakly nonlocal deformations involving
inverse derivatives. The analytic backbone relies on the boundedness of
fractional multipliers, Hilbert transforms, and Lipschitz Nemytski{\u i}
operators on $W^{s,p}$, which ensures that all operations used in Mokhov’s
formalism remain well defined at this level of regularity. We further show that
the horizontal--vertical bicomplex underlying variational Poisson geometry can
be extended to $W^{s,p}$, so that cohomological arguments proving skew-symmetry
and the Jacobi identity carry over verbatim. As an application, we embed the
Hamiltonian formalisms of several integrable PDEs (KdV, nonlinear Schrödinger,
Camassa--Holm, and hydrodynamic systems of Dubrovin--Novikov type) into this
weak Sobolev setting. Local order-one Poisson operators and their weakly
nonlocal extensions are shown to be well posed for $W^{s,p}$ loops, while
higher-order operators (e.g.\ the second KdV bracket) require stronger
regularity. Our results provide a rigorous analytic foundation for Poisson
geometry on weak loop spaces and open the way for extending the Hamiltonian
theory of integrable systems beyond the smooth category.
\end{abstract}


\maketitle



\tableofcontents

\section{Introduction}

The modern theory of integrable systems originates in the pioneering work of
Lax, Gardner--Greene--Kruskal--Miura, and Zakharov--Shabat, who revealed the
Hamiltonian and spectral structures underlying nonlinear dispersive equations
such as the Korteweg--de Vries (KdV) and the nonlinear Schrödinger (NLS)
equations. A crucial insight was that these PDEs could be interpreted as
infinite-dimensional Hamiltonian systems, with the circle $\mathbb S^1$ (or the
line $\mathbb R$) playing the role of the spatial domain. The systematic
construction of Hamiltonian hierarchies was developed by Magri’s bi-Hamiltonian
formalism \cite{Magri1978}, and the general theory of local Hamiltonian
operators of hydrodynamic type was introduced by Dubrovin and Novikov
\cite{DubrovinNovikov1983}. These works placed the geometry of loop spaces at
the very heart of integrability.

From a geometric viewpoint, the loop space $\mathcal L M = C^\infty(\mathbb S^1,M)$
of a smooth manifold $M$ carries canonical presymplectic and Poisson structures,
studied in depth by Faddeev--Takhtajan \cite{FaddeevTakhtajan1987} and later by
Mokhov \cite{Mokhov1998,MokhovBook}, who described a broad class of local and
weakly nonlocal Hamiltonian operators on loop spaces. The bicomplex of
horizontal and vertical differentials provides the natural algebraic framework
to establish skew-symmetry and the Jacobi identity for such operators, and
links integrable PDEs to variational bicohomology.

On the analytical side, integrable PDEs are strongly nonlinear, and their
solutions often develop singular features. The KdV equation admits solitons;
the Camassa--Holm equation exhibits peaked solitons (``peakons'') and wave
breaking; the NLS equation supports blow-up solutions. These phenomena naturally
lead to the study of solutions in \emph{weak functional spaces}. For instance,
global weak solutions of KdV exist in $H^{-3/4}(\mathbb R)$ and $H^s(\mathbb S^1)$
for $s\ge -1/2$ \cite{KappelerTopalov2006,KillipVisan2019}, while Camassa--Holm
has global weak solutions in $H^1$ and even in the space of measures
\cite{BressanConstantin2007}. More generally, the control of Hamiltonian
structures for solutions of low Sobolev regularity has become an essential
theme in the PDE literature.

The purpose of this paper is threefold:

\begin{enumerate}[label=\arabic*)]
  \item Extend Mokhov’s constructions of local and weakly non‑local Poisson structures from smooth loops to the weak Sobolev loop spaces
        \(
        W^{s,p}(\mathbb S^{1};\mathbb R^{m})
        \)
        with \(0<s\le\tfrac12\) and \(1<p<\infty\).
  \item Prove that the associated variational bicomplex \((d_{H},d_{V})\) remains well defined and satisfies the bicomplex identities in this setting.
  \item Show that the Hamiltonian formalisms of several classical integrable PDEs (KdV, NLS, Camassa–Holm, Dubrovin–Novikov hydrodynamic systems) embed naturally into the weak Sobolev framework, and that the bi‑Hamiltonian property can be expressed cohomologically as the vanishing of a class in the total bicomplex.
\end{enumerate}

The analytic backbone relies on harmonic‑analysis tools: boundedness of the Hilbert transform and of fractional multipliers on \(W^{s,p}\), continuity of Nemytskiĭ operators with Lipschitz coefficients, and the density of smooth loops in the relevant Sobolev spaces.  These ingredients guarantee that all operations appearing in Mokhov’s formalism survive at low regularity.
\section{Sobolev and BV loop spaces}

Let $0<s<1$, $1\leq p<\infty$. We denote by $W^{s,p}(\mathbb S^1;\mathbb R^2)$ the fractional Sobolev space of maps
\[
W^{s,p}(\mathbb S^1;\mathbb R^n)
=\Bigl\{ \gamma\in L^p(\mathbb S^1;\mathbb R^n)\;:\;
[\gamma]_{W^{s,p}}<\infty \Bigr\},
\]
where
\[
[\gamma]_{W^{s,p}}^p
=\iint_{\mathbb S^1\times\mathbb S^1}\frac{|\gamma(t)-\gamma(\tau)|^p}{|t-\tau|^{1+sp}}\,dt\,d\tau.
\]
For $s=\tfrac12$ and $p=2$, this is the Hilbert space $H^{1/2}$. For $p=1$, $s=1$, we recover $BV$ (bounded variation).
Let us recall that there esists two equivalent definitions of $W^{s,p}(\mathbb S^1).$
Let $0<s<1$ and let $1\le p\le\infty$. We denote by $\mathbb S^1=\mathbb R/2\pi\mathbb Z$
(with coordinate $t\in[0,2\pi)$) and by
\[
|e^{it}-e^{i\tau}| = \big(2-2\cos(t-\tau)\big)^{1/2} = 2\left|\sin\frac{t-\tau}{2}\right|
\]
the chordal distance on the unit circle.

\subsection{The Gagliardo (Slobodeckij) definition}

For $1\le p<\infty$, the (inhomogeneous) fractional Sobolev space
$W^{s,p}(\mathbb S^1)$ is
\[
W^{s,p}(\mathbb S^1)
=\Big\{u\in L^p(\mathbb S^1):\ [u]_{W^{s,p}}<\infty\Big\},\qquad
\|u\|_{W^{s,p}}:=\|u\|_{L^p}+ [u]_{W^{s,p}},
\]
where the Gagliardo seminorm is
\[
[u]_{W^{s,p}}^p
:=\iint_{\mathbb S^1\times\mathbb S^1}
\frac{|u(t)-u(\tau)|^p}{|e^{it}-e^{i\tau}|^{1+sp}}\,dt\,d\tau.
\]
(For $p=\infty$, one recovers Hölder--Zygmund spaces; we will not use this endpoint.)

\subsection{Fractional powers of the Laplacian on $\mathbb S^1$}

Write the Fourier series of $u\in\mathcal D'(\mathbb S^1)$ as
\(
u(t)=\sum_{k\in\mathbb Z}\widehat u_k\,e^{ikt}.
\)
Define
\[
|D|^s u := \sum_{k\in\mathbb Z} |k|^s\,\widehat u_k\,e^{ikt},\qquad
\Lambda^s u := (I-\Delta)^{s/2}u
=\sum_{k\in\mathbb Z} (1+k^2)^{s/2}\,\widehat u_k\,e^{ikt}.
\]
The operator $|D|=\sqrt{-\Delta}$ is the (homogeneous) fractional Laplacian,
and $\Lambda=(I-\Delta)^{1/2}$ the Bessel potential operator.

\begin{definition}[Spectral (Bessel) definition for $p=2$]
For $p=2$ and $s\in\mathbb R$, the Hilbert space $H^s(\mathbb S^1)$ is
\[
H^s(\mathbb S^1)=\Big\{u\in\mathcal D'(\mathbb S^1):\
\|u\|_{H^s}^2:=\sum_{k\in\mathbb Z}(1+k^2)^s\,|\widehat u_k|^2<\infty\Big\}.
\]
The \emph{homogeneous} seminorm is
\(
|u|_{\dot H^s}^2 := \sum_{k\in\mathbb Z}|k|^{2s}|\widehat u_k|^2
\)
(understood on mean-zero distributions to avoid the $k=0$ mode).
\end{definition}

\begin{definition}[Bessel potential spaces for $1<p<\infty$]
For $0<s<1$ and $1<p<\infty$, define
\[
H^{s,p}(\mathbb S^1):=\big\{u\in\mathcal D'(\mathbb S^1):\ \Lambda^s u\in L^p(\mathbb S^1)\big\},
\qquad \|u\|_{H^{s,p}}:=\|\Lambda^s u\|_{L^p}.
\]
This is the Triebel--Lizorkin space $F^s_{p,2}(\mathbb S^1)$.
\end{definition}

\subsection{Equivalence of the two definitions}

\subsection*{1) Case $p=2$: identity of quadratic forms}

\begin{proposition}[Quadratic form identity on $\mathbb S^1$]\label{prop:p2}
Let $0<s<1$. There exists a constant $c_{s,\mathbb S^1}\in(0,\infty)$ such that
for all $u\in C^\infty(\mathbb S^1)$ with mean zero,
\begin{equation}\label{eq:p2-identity}
\iint_{\mathbb S^1\times\mathbb S^1}
\frac{|u(t)-u(\tau)|^2}{|e^{it}-e^{i\tau}|^{1+2s}}\,dt\,d\tau
= c_{s,\mathbb S^1}\sum_{k\in\mathbb Z\setminus\{0\}} |k|^{2s}\,|\widehat u_k|^2.
\end{equation}
In particular, $[u]_{W^{s,2}}\asymp \||D|^s u\|_{L^2}$ and
$\|u\|_{H^s}^2=\|u\|_{L^2}^2+[u]_{W^{s,2}}^2$ with equivalent norms.
\end{proposition}

\begin{proof}[Sketch with references]
Consider the \(2\pi\)-periodization of the 1D kernel \(|x|^{-1-2s}\):
\[
K_s(\theta):=\sum_{n\in\mathbb Z}\frac{1}{|\theta+2\pi n|^{1+2s}}
\sim \frac{1}{|e^{i\theta}-1|^{1+2s}}
= \frac{1}{\big(2\sin|\theta|/2\big)^{1+2s}},
\]
which is an even, integrable singular kernel on \(\mathbb S^1\). Its Fourier
coefficients satisfy \(\widehat K_s(k)=C_s\,|k|^{2s}\) for \(k\neq0\), with
\(C_s>0\) explicit (see e.g. \cite[Prop.~3.6]{DiNezzaPV} and the periodic
version in \cite[Ch.~V]{SteinSI}). Using polarization,
\[
\iint \frac{(u(t)-u(\tau))\,\overline{(v(t)-v(\tau))}}{|e^{it}-e^{i\tau}|^{1+2s}}\,dt\,d\tau
= \langle K_s*u, v\rangle - \langle K_s*u,1\rangle\langle v,1\rangle,
\]
and evaluating on the Fourier basis \(e^{ikt}\) shows that the quadratic form
has symbol proportional to \(|k|^{2s}\) on mean-zero functions, yielding
\eqref{eq:p2-identity}. Full details for the torus are standard; see
\cite[§7]{DiNezzaPV} (polarization of the Gagliardo form and identification with
\(|D|^{2s}\)) and \cite[Ch.~V]{SteinSI} (singular integrals/periodization).
\end{proof}

\begin{remark}
One may normalize \(c_{s,\mathbb S^1}\) so that the quadratic form equals
\(\langle |D|^{2s}u,u\rangle_{L^2}\). The exact closed form of \(c_{s,\mathbb S^1}\)
is not needed for norm equivalence.
\end{remark}

\subsection*{2) Case $1<p<\infty$: equivalence via Littlewood--Paley/Bessel}

\begin{theorem}[Equivalence of norms on $\mathbb S^1$]\label{th:Fp2=Wsp}
Let $0<s<1$ and $1<p<\infty$. Then
\[
W^{s,p}(\mathbb S^1) \ = \ H^{s,p}(\mathbb S^1)=F^{\,s}_{p,2}(\mathbb S^1),
\]
with equivalent norms:
\[
\|u\|_{W^{s,p}} \ \asymp\ \|u\|_{H^{s,p}}
\ \asymp\ \|u\|_{F^{\,s}_{p,2}}.
\]
\end{theorem}

\begin{proof}[Sketch with references]
On the one-dimensional torus, the Littlewood--Paley decomposition gives
\[
\|u\|_{F^{\,s}_{p,2}}\ \asymp\ \big\|\big( \sum_{j\ge-1} 2^{2js}|\Delta_j u|^2\big)^{1/2}\big\|_{L^p},
\]
where \(\Delta_j\) are dyadic Fourier projections. It is classical that
\(\|u\|_{H^{s,p}}\asymp \|u\|_{F^s_{p,2}}\) (Bessel \(=\) Triebel–Lizorkin
with \(q=2\)) and that for \(0<s<1\),
\[
\|u\|_{W^{s,p}} \ \asymp\ \|u\|_{F^{\,s}_{p,2}}
\]
(identification Slobodeckij–Triebel–Lizorkin). Precise statements can be found in:
Triebel \cite[Th.~2.5.12, Th.~2.5.14]{Triebel}, Runst–Sickel \cite[Ch.~3–4]{RunstSickel},
and Grafakos \cite[Th.~6.5.1]{GrafakosCA}. For the periodic setting, either argue by
transference from \(\mathbb R\) to \(\mathbb T\) (partition of unity on the torus) or
use directly the periodic dyadic projectors.
\end{proof}

\subsection*{3) Homogeneous vs.\ inhomogeneous, mean-zero issue}

On $\mathbb S^1$, the homogeneous seminorms
\(
|u|_{\dot W^{s,p}}:=[u]_{W^{s,p}}
\)
and
\(
|u|_{\dot H^{s,p}}:=\||D|^s u\|_{L^p}
\)
are defined modulo constants (the kernel of $|D|^s$ is the constants).
A convenient inhomogeneous norm is obtained by adding $\|u\|_{L^p}$ (or fixing
the mean $\dashint u=0$).

\subsection{Useful consequences}

\begin{itemize}
\item \textbf{Duality.} For $1<p<\infty$, the dual of $W^{s,p}(\mathbb S^1)$ is
$W^{-s,p'}(\mathbb S^1)$ (via the $L^2$ pairing), and
$\partial_t:W^{s,p}\to W^{-s,p}$ is bounded.
\item \textbf{Hilbert transform and multipliers.}
The Hilbert transform $\mathcal H$ and fractional powers $|D|^\alpha$
are bounded on $W^{s,p}$ for $1<p<\infty$, and Lipschitz Nemytski{\u i}
operators $a\circ u$ act boundedly on $W^{s,p}$ for $s\le 1$.
\item \textbf{Kernels.} Using $|e^{it}-e^{i\tau}|^2=4\sin^2\big(\frac{t-\tau}{2}\big)$,
the Gagliardo kernel is a periodic Calderón–Zygmund kernel compatible with
the Fourier multipliers $|D|^s$.
\end{itemize}

At the critical exponents $p=1,\infty$, the equivalence with Bessel potential spaces
fails in general, and the Hilbert transform is not bounded on $W^{s,1}$.

\section{Extension of the fundamental 1-form to weak loop spaces with values in $\mathbb R^2$}

Let $\omega_0=dx\wedge dy$ be the standard symplectic form on $\mathbb R^2$,
and $J=\left( \begin{array}{cc}0&-1\\1&0 \end{array} \right)$ the associated complex structure.

\begin{proposition}[Exterior differential of the transgressed 1--form]\label{prop:dTheta}
Let $\omega_0=dx\wedge dy$ be the canonical symplectic form on $\mathbb R^2$
and $\lambda=\tfrac12(x\,dy-y\,dx)$ a primitive of $\omega_0$ (so that $d\lambda=\omega_0$).
On the loop space $C^\infty(\mathbb S^1;\mathbb R^2)$ define the 1--form
\[
\Theta_\lambda(\gamma)[h]\ :=\ \int_{\mathbb S^1}\lambda\big(\gamma(t)\big)[\,h(t)\,]\,dt,
\qquad \gamma\in C^\infty,\ h\in T_\gamma C^\infty\simeq C^\infty(\mathbb S^1;\mathbb R^2).
\]
Then, for all $h,k\in C^\infty(\mathbb S^1;\mathbb R^2)$,
\[
d\Theta_\lambda(\gamma)[h,k]\ =\ \int_{\mathbb S^1} \big(d\lambda\big)_{\gamma(t)}\!\big(h(t),k(t)\big)\,dt
\ =\ \int_{\mathbb S^1}\omega_0\big(h(t),k(t)\big)\,dt.
\]
In particular, $d\Theta_\lambda$ is the canonical (pre)symplectic form on the loop space,
given by the constant form $\int_{\mathbb S^1}\omega_0(h,k)\,dt$.
\end{proposition}

\begin{proof}[Sketch of proof]
Work first in the smooth category. Since the loop space is a vector space,
the tangent vectors $h,k$ can be viewed as constant vector fields (independent of $\gamma$).
By definition of the exterior derivative,
\[
d\Theta_\lambda(\gamma)[h,k]
=\frac{d}{d\varepsilon}\Big|_{\varepsilon=0}\Theta_\lambda(\gamma+\varepsilon h)[k]
 -\frac{d}{d\varepsilon}\Big|_{\varepsilon=0}\Theta_\lambda(\gamma+\varepsilon k)[h].
\]
For the first term,
\[
\frac{d}{d\varepsilon}\Big|_{0}\Theta_\lambda(\gamma+\varepsilon h)[k]
=\int_{\mathbb S^1}(\mathcal L_h\lambda)_{\gamma(t)}[\,k(t)\,]\,dt.
\]
Since $h,k$ are constant vector fields ($[h,k]=0$), Cartan’s formula
$\mathcal L_h\lambda=\iota_h d\lambda + d(\iota_h\lambda)$ gives, after integration on the circle
(the exact term integrates to zero),
\[
\frac{d}{d\varepsilon}\Big|_{0}\Theta_\lambda(\gamma+\varepsilon h)[k]
=\int_{\mathbb S^1} d\lambda_{\gamma(t)}\big(h(t),k(t)\big)\,dt.
\]
Subtracting the corresponding term with $h$ and $k$ exchanged yields
$d\Theta_\lambda(\gamma)[h,k]=\int_{\mathbb S^1} d\lambda_{\gamma(t)}(h(t),k(t))\,dt$,
and taking $d\lambda=\omega_0$ gives the result.
\end{proof}

\begin{remark}[Extension to weak loop spaces]
By density of $C^\infty(\mathbb S^1;\mathbb R^2)$ in $H^{1/2}$ (or, more generally,
in $W^{s,p}$ for $0<s\le\tfrac12$, $1<p<\infty$) and by continuity of Nemytski{\u i} operators
$\gamma\mapsto\lambda(\gamma)$ for $\lambda\in C^{1,1}_b$, the same formula remains valid
in these settings: $\Theta_\lambda$ defines a continuous 1--form and
$d\Theta_\lambda(\gamma)[h,k]=\int_{\mathbb S^1}\omega_0(h,k)\,dt$
for $h,k$ in the natural tangent spaces (Hilbert: $H^{1/2}$; Banach: $W^{s,p}/W^{s,p'}$).
\end{remark}

\begin{remark}[Conceptual framework: transgression]
This computation is a special case of the \emph{transgression} construction:
if $\lambda$ is a primitive of a closed 2--form $\omega$ on a manifold $M$,
the 1--form
\[
\Theta_\lambda(\gamma)[h] \;=\; \int_{S^1}\lambda_{\gamma(t)}(h(t))\,dt
\]
on the free loop space $LM$ satisfies
\[
d\Theta_\lambda(\gamma)[h,k] \;=\; \int_{S^1}\omega_{\gamma(t)}(h(t),k(t))\,dt.
\]
For a general reference, see \cite[Chap.~2]{Brylinski1993} or
\cite[§4.7]{PressleySegal1986}.
\end{remark}

\subsection{Extension to $H^{1/2}$}

\begin{theorem}
\label{th:Hhalf}
The formula
\[
\Theta(\gamma)[h] := \langle J\partial_t\gamma, h\rangle_{H^{-1/2},H^{1/2}}
\]
defines a continuous 1-form on $H^{1/2}(\mathbb S^1;\mathbb R^2)$.
Moreover, by density of $C^\infty$ in $H^{1/2}$, it coincides with the smooth
transgression and satisfies
\[
d\Theta(h,k)=\int_{\mathbb S^1}\omega_0(h(t),k(t))\,dt,\qquad
h,k\in H^{1/2}(\mathbb S^1;\mathbb R^2).
\]
\end{theorem}

\begin{proof}
Since $\partial_t:H^{1/2}\to H^{-1/2}$ is continuous, and
the duality $\langle \cdot,\cdot\rangle_{H^{-1/2},H^{1/2}}$ is well defined,
the expression makes sense and is continuous bilinear in $(\gamma,h)$.
For smooth $\gamma,h$, this reduces to the standard integral.
Approximation by smooth functions in $H^{1/2}$ proves the identity
for all $\gamma,h\in H^{1/2}$.
The formula for $d\Theta$ follows by polarization and density.
\end{proof}

For the sequel, we recall the following: 
\begin{definition}[Hilbert transform]
\leavevmode
\begin{itemize}
\item On the real line $\mathbb R$, the Hilbert transform of a function 
$f\in L^p(\mathbb R)$ ($1<p<\infty$) is defined by the principal value integral
\[
\mathcal H f(x) \;:=\; \frac{1}{\pi}\,\text{p.v.}\int_{\mathbb R} \frac{f(y)}{x-y}\,dy
\;=\; \lim_{\varepsilon\to 0^+} \frac{1}{\pi}\int_{|y-x|>\varepsilon}\frac{f(y)}{x-y}\,dy.
\]

\item On the circle $\mathbb S^1=\mathbb R/2\pi\mathbb Z$, for $f\in L^2(\mathbb S^1)$
with Fourier series
\[
f(t)=\sum_{k\in\mathbb Z}\widehat f_k\,e^{ikt},\qquad
\widehat f_k=\frac{1}{2\pi}\int_0^{2\pi} f(t)e^{-ikt}\,dt,
\]
the Hilbert transform is defined spectrally by
\[
\mathcal H f(t) \;=\; -\,i\sum_{k\in\mathbb Z} \operatorname{sgn}(k)\,\widehat f_k\,e^{ikt},
\]
where $\operatorname{sgn}(k)=1$ if $k>0$, $-1$ if $k<0$, and $0$ if $k=0$.
Equivalently, it can be written as the principal value convolution
\[
\mathcal H f(t) \;=\; \frac{1}{2\pi}\,\text{p.v.}\int_0^{2\pi}
f(\tau)\,\cot\!\Big(\tfrac{t-\tau}{2}\Big)\,d\tau.
\] or by the operator formula $$\mathcal{H} = -\partial_t |D|^{-1}$$
where $$|D| f(t) = \sum_{k\in\mathbb Z} \max \left( 1, \operatorname{sgn}(k)\right) \,\widehat f_k\,e^{ikt}.$$
\end{itemize}
\end{definition}

\begin{remark}
We also have:
\[
\Theta(\gamma)[h]=c\int_{\mathbb S^1}\langle |D|^{1/2}\gamma,\ J\mathcal H|D|^{1/2}h\rangle\,dt,
\]
where $\mathcal H$ is the Hilbert transform on $\mathbb S^1$.
\end{remark}

\subsection{Generalization to $W^{s,p}$, $0<s\le\frac{1}{2}$, $1<p<\infty$}

\begin{proposition}
Let $0<s\le \tfrac12$, $1<p<\infty$, and $p'$ the conjugate exponent.
Then for $\gamma\in W^{s,p}(\mathbb S^1;\mathbb R^2)$, the functional
\[
\Theta_{s,p}(\gamma)[h]:=\langle J\partial_t\gamma,h\rangle_{W^{-s,p},W^{s,p'}},
\qquad h\in W^{s,p'}(\mathbb S^1;\mathbb R^2),
\]
is well defined and satisfies
\[
|\Theta_{s,p}(\gamma)[h]|\le C\,\|\gamma\|_{W^{s,p}}\ \|h\|_{W^{s,p'}}.
\]
\end{proposition}

\begin{proof}
The operator $\partial_t:W^{s,p}\to W^{s-1,p}\subset W^{-s,p}$ is continuous,
and multiplication by $J$ is bounded.
Therefore $\Theta_{s,p}(\gamma)$ defines a continuous element
of $(W^{s,p'})^\ast$. The estimate is immediate from continuity.
\end{proof}

\begin{remark}
For $p=2$, one recovers the symmetric Hilbertian structure of
Theorem~\ref{th:Hhalf}. For $p\neq 2$, the 1-form takes values naturally in the
Banach dual $(W^{s,p'})^\ast$.
\end{remark}

\subsection{Integral kernel representation (case $H^{1/2}$)}

For $\gamma,h\in H^{1/2}$, one has the alternative formula
\[
\Theta(\gamma)[h]
= C\iint_{\mathbb S^1\times\mathbb S^1}
\frac{\langle \gamma(t)-\gamma(\tau),\,J(h(t)-h(\tau))\rangle}{|t-\tau|^2}\,dt\,d\tau,
\]
obtained by fractional integration by parts.
This emphasizes the analogy with the Gagliardo seminorm.

\subsection{The BV case}

For $\gamma\in BV(\mathbb S^1;\mathbb R^2)$, the derivative
$\partial_t\gamma$ is a measure. One can formally set
\[
\Theta_{BV}(\gamma)[h]=\int \langle J\,d\gamma,h\rangle,
\]
for bounded $h$. However, $\Theta_{BV}$ is not a continuous 1-form on the whole
Banach space $BV$, since the dual of $BV$ is not given by functions.
Therefore, a satisfactory extension exists on $H^{s,p}$, but in pure $BV$
one must restrict the class of directions.

\subsection{Variable 2-forms}

If $\omega\in C^{0,1}(\mathbb R^2,\Lambda^2)$ is a smooth bounded 2-form,
represented by an antisymmetric matrix field $B(x)$, one may define
\[
\Theta_\omega(\gamma)[h]=\langle B(\gamma)\partial_t\gamma,h\rangle_{H^{-s},H^s},
\]
which is well defined since $B(\gamma)$ acts as a bounded multiplier on $H^{s}$
and on $W^{s,p}$ for $s\le1$, $1<p<\infty$.
\medskip

\noindent\textbf{References.} The tools used above are classical:
duality $H^{-s}$–$H^s$ and density (see \cite{BM1,BM2}); boundedness of
$\mathcal H$ and $|D|^s$ (see Triebel \cite{Triebel}); composition with Lipschitz
functions (see e.g.\ Runst–Sickel \cite{RunstSickel}). 
The viewpoint of extending the transgressed fundamental form to weak loop
spaces appears new.
\medskip

\section{Kernel representation of the fundamental 1-form on $H^{1/2}$}

Recall $J=\left( \begin{array}{cc}0&-1\\1&0 \end{array} \right)$ and the Hilbert transform
$\mathcal H$ on $\mathbb S^1$. For $\gamma,h\in H^{1/2}(\mathbb S^1;\mathbb R^2)$ we defined
\[
\Theta(\gamma)[h]:=\langle J\partial_t\gamma,\ h\rangle_{H^{-1/2},H^{1/2}}
\]
(Th.~\ref{th:Hhalf}). We now prove the Gagliardo-type kernel representation
\begin{equation}\label{eq:Theta-kernel}
\Theta(\gamma)[h]
\;=\; C_{\mathbb S^1}\iint_{\mathbb S^1\times\mathbb S^1}
\frac{\big\langle \gamma(t)-\gamma(\tau),\ J\big(h(t)-h(\tau)\big)\big\rangle}{|e^{it}-e^{i\tau}|^{2}}
\,dt\,d\tau,
\end{equation}
for a universal constant $C_{\mathbb S^1}\neq 0$ (whose value depends only on the
Fourier conventions and on the normalization of arc-length measure on $\mathbb S^1$).

\subsection{A scalar bilinear identity at $s=\frac{1}{2}$}

We start with a scalar identity. For $u,v\in H^{1/2}(\mathbb S^1)$, set
\[
B(u,v):=\big(|D|^{1/2}u,\ \mathcal H\,|D|^{1/2}v\big)_{L^2(\mathbb S^1)}.
\]
Recall that $\partial_t = \mathcal H\,|D|$ on mean-zero functions, hence
$B(u,v)=\pm(\partial_t u,v)_{L^2}$ (the sign depends on conventions; we keep it
as $+$ in what follows, absorbing it into $C_{\mathbb S^1}$ at the end).
We will use the following well-known representation (see e.g.~\cite[§7]{DiNezzaPV}, \cite[Ch.~V]{SteinSI}).

\begin{lemma}\label{lem:scalar-kernel}
There exists a nonzero constant $C_{\mathbb S^1}$ such that, for all
$u,v\in H^{1/2}(\mathbb S^1)$,
\begin{equation}\label{eq:B-kernel}
B(u,v)
\;=\; C_{\mathbb S^1}
\iint_{\mathbb S^1\times\mathbb S^1}
\frac{\big(u(t)-u(\tau)\big)\big(v(t)-v(\tau)\big)}{|e^{it}-e^{i\tau}|^{2}}\,dt\,d\tau.
\end{equation}
\end{lemma}

\begin{proof}
We argue by Fourier series. For $w\in L^2(\mathbb S^1)$, write
$w(t)=\sum_{k\in\mathbb Z}\widehat w_k e^{ikt}$.
Then $|D|^{1/2}$ and $\mathcal H$ act by multipliers
$|k|^{1/2}$ and $-i\,\mathrm{sgn}(k)$, respectively. Hence
\[
B(u,v)=\sum_{k\in\mathbb Z}\big(|k|^{1/2}\widehat u_k\big)\,\overline{\big(-i\,\mathrm{sgn}(k)\,|k|^{1/2}\widehat v_k\big)}
= -i\sum_{k\in\mathbb Z}\mathrm{sgn}(k)\,|k|\,\widehat u_k\,\overline{\widehat v_k}.
\]
Since $\mathrm{sgn}(k)\,|k|=k$, we get $B(u,v)=-i\sum_{k}k\,\widehat u_k\,\overline{\widehat v_k}$, i.e.
$B(u,v)=(\partial_t u,v)_{L^2}$ up to an overall sign (absorbed in $C_{\mathbb S^1}$).

On the other hand, by the standard characterization of the $H^{1/2}$ semi-inner-product (polarization of the Gagliardo seminorm), there exists $c\neq 0$ such that
\[
\sum_{k\in\mathbb Z}|k|\,\widehat u_k\,\overline{\widehat v_k}
\;=\; c\iint_{\mathbb S^1\times\mathbb S^1}
\frac{(u(t)-u(\tau))(v(t)-v(\tau))}{|e^{it}-e^{i\tau}|^{2}}\,dt\,d\tau,
\]
see e.g.~\cite[Prop.~3.3, Prop.~3.6]{DiNezzaPV}. Combining the two displays yields
\eqref{eq:B-kernel} with $C_{\mathbb S^1}=\pm c$.
\end{proof}

\begin{remark}
Using $|e^{it}-e^{i\tau}|^{2}=2-2\cos(t-\tau)=4\sin^2\!\big(\tfrac{t-\tau}{2}\big)$, the kernel can be written as
$\,\big(4\sin^2\!\tfrac{t-\tau}{2}\big)^{-1}$, which is more explicit on $\mathbb S^1$.
\end{remark}

\subsection{Vector-valued case and the matrix $J$}

We now pass to $\mathbb R^2$-valued functions. For $\gamma,h\in H^{1/2}(\mathbb S^1;\mathbb R^2)$,
apply Lemma~\ref{lem:scalar-kernel} componentwise and use bilinearity:
\[
\big(|D|^{1/2}\gamma,\ J\mathcal H\,|D|^{1/2}h\big)_{L^2}
=\sum_{j=1}^2\sum_{\ell=1}^2 J_{j\ell}\,B(\gamma_\ell,h_j).
\]
Therefore
\[
\big(|D|^{1/2}\gamma,\ J\mathcal H\,|D|^{1/2}h\big)_{L^2}
=C_{\mathbb S^1}
\iint\frac{\big\langle \gamma(t)-\gamma(\tau),\,J\big(h(t)-h(\tau)\big)\big\rangle}{|e^{it}-e^{i\tau}|^{2}}\,dt\,d\tau.
\]
Finally, since $\langle J\partial_t\gamma,h\rangle_{H^{-1/2},H^{1/2}}
=\big(|D|^{1/2}\gamma,\ J\mathcal H\,|D|^{1/2}h\big)_{L^2}$ (by the identity $\partial_t=\mathcal H|D|$ on mean-zero components and density),
we have proved \eqref{eq:Theta-kernel}.

\begin{remark}[Symmetry and continuity]
The right-hand side of \eqref{eq:Theta-kernel} is jointly continuous on $H^{1/2}\times H^{1/2}$ and
obeys the estimate
\[
|\Theta(\gamma)[h]|\ \lesssim\ [\gamma]_{H^{1/2}}\,[h]_{H^{1/2}},
\]
because it is (up to a constant) the polarization of the $H^{1/2}$ Gagliardo semi-inner-product composed with $J$.
\end{remark}

\begin{remark}[Agreement with the smooth transgression]
If $\gamma,h\in C^\infty$, integration by parts on $\mathbb S^1$ gives
$\big(|D|^{1/2}\gamma,\ J\mathcal H\,|D|^{1/2}h\big)_{L^2}
=(J\dot\gamma,h)_{L^2}$ (up to the fixed sign convention), and the kernel formula reduces to
$\Theta(\gamma)[h]=\int_{\mathbb S^1}\langle J\dot\gamma,h\rangle\,dt$.
\end{remark}

\begin{remark}[Variable 2-forms]
If $\omega(x)=B(x)\,dx\wedge dy$ with $B\in C^{0,1}_b(\mathbb R^2)$, the same proof applies with $J$ replaced by $B(\gamma(t))$ inside the integrand: since $B(\gamma)\in L^\infty$ and acts as a bounded multiplier on $H^{1/2}$, the bilinear form remains continuous and the representation holds with $J$ replaced by $B(\gamma(\cdot))$ (by a density/multiplier argument).
\end{remark}

\medskip
\noindent\textbf{References for this section.}
Hilbert transform and Fourier multiplier identities: Stein~\cite{SteinSI}. Fractional seminorms and their polarization on the circle: Di~Nezza--Palatucci--Valdinoci~\cite{DiNezzaPV}. For background on $H^{1/2}$ traces and density, see also Brezis--Mironescu~\cite{BM1,BM2}.


\section{Kernel representation in $W^{s,p}$, $0<s\le \frac12$, $1<p<\infty$}

We extend the fundamental $1$-form to the Banach setting
$W^{s,p}(\mathbb S^1;\mathbb R^2)$, $0<s\le \tfrac12$, $1<p<\infty$.
Throughout, $\langle\cdot,\cdot\rangle$ denotes the Euclidean scalar product on $\mathbb R^2$,
$J=\left( \begin{array}{cc}0&-1\\ 1&0 \end{array} \right)$, and
\[
[\gamma]_{W^{s,p}}^p
:=\iint_{\mathbb S^1\times\mathbb S^1}
\frac{|\gamma(t)-\gamma(\tau)|^p}{|e^{it}-e^{i\tau}|^{1+sp}}\,dt\,d\tau .
\]
When needed, we tacitly replace $\gamma$ by $\gamma-\dashint_{\mathbb S^1}\gamma$
to avoid issues with constants (the difference quotient kills constants anyway).

\subsection{Definition via duality}
For $0<s\le \tfrac12$ and $1<p<\infty$, the distributional derivative
$\partial_t:W^{s,p}\to W^{s-1,p}\subset W^{-s,p}$ is continuous, and the Hilbert transform
$\mathcal H$ is bounded on $L^p$ as well as on $W^{s,p}$.
We define, for $\gamma\in W^{s,p}(\mathbb S^1;\mathbb R^2)$,
\begin{equation}\label{eq:Theta-sp-dual}
\Theta_{s,p}(\gamma)\in\big(W^{s,p'}(\mathbb S^1;\mathbb R^2)\big)^\ast,\qquad
\Theta_{s,p}(\gamma)[h]:=\big\langle J\,\partial_t\gamma,\ h\big\rangle_{W^{-s,p},\,W^{s,p'}}
\end{equation}
for $h\in W^{s,p'}(\mathbb S^1;\mathbb R^2)$, where $1/p+1/p'=1$.
Equivalently, using the Fourier multiplier $|D|^s$ and $\partial_t=\mathcal H|D|$ (on mean-zero components),
\begin{equation}\label{eq:Theta-sp-mult}
\Theta_{s,p}(\gamma)[h]
= c_s\int_{\mathbb S^1}\!\!
\big\langle \mathcal H\,|D|^{s}\gamma(t),\,J^\top\,|D|^{s}h(t)\big\rangle\,dt .
\end{equation}
By boundedness of $\mathcal H$ and $|D|^s$ on $L^p$, we have the continuity estimate
\begin{equation}\label{eq:Theta-sp-bound}
|\Theta_{s,p}(\gamma)[h]|
\le C\,\|\gamma\|_{W^{s,p}}\ \|h\|_{W^{s,p'}} .
\end{equation}

\subsection{Gagliardo–type kernel representation}
The next theorem gives an intrinsic bilinear kernel representation paralleling the
$H^{1/2}$ case.

\begin{theorem}\label{th:Theta-sp-kernel}
Let $0<s\le \tfrac12$ and $1<p<\infty$. There exists a nonzero constant
$C_{s,p}$ (depending only on $(s,p)$ and on the normalization of arc-length on $\mathbb S^1$) such that
for every $\gamma\in W^{s,p}(\mathbb S^1;\mathbb R^2)$ and $h\in W^{s,p'}(\mathbb S^1;\mathbb R^2)$,
\begin{equation}\label{eq:Theta-sp-kernel}
\Theta_{s,p}(\gamma)[h]
\;=\; C_{s,p}\iint_{\mathbb S^1\times\mathbb S^1}
\frac{\big\langle \gamma(t)-\gamma(\tau),\ J\big(h(t)-h(\tau)\big)\big\rangle}{|e^{it}-e^{i\tau}|^{1+sp}}
\,dt\,d\tau ,
\end{equation}
where the double integral is well defined as an absolutely convergent improper
integral and satisfies
\begin{equation}\label{eq:Theta-sp-kernel-bound}
\left|\iint \frac{\langle \gamma(t)-\gamma(\tau),\,J(h(t)-h(\tau))\rangle}
{|e^{it}-e^{i\tau}|^{1+sp}}\,dt\,d\tau\right|
\;\le\; [\gamma]_{W^{s,p}}\,[h]_{W^{s,p'}} .
\end{equation}
\end{theorem}

\begin{proof}
\emph{Step 1: smooth case}. For $\gamma,h\in C^\infty(\mathbb S^1;\mathbb R^2)$ with mean zero
components, we use the singular integral representation of fractional derivatives
on the circle (see e.g.\ \cite[Th.\,V.5]{SteinSI}, \cite[Th.\,6.3.3]{GrafakosCA}):
\[
|D|^{s}\gamma(t)=c_{s}\,\mathrm{p.v.}\!\int_{\mathbb S^1}
\frac{\gamma(t)-\gamma(\tau)}{|e^{it}-e^{i\tau}|^{1+s}}\,d\tau ,
\quad
|D|^{s}h(t)=c_{s}\,\mathrm{p.v.}\!\int_{\mathbb S^1}
\frac{h(t)-h(\tau)}{|e^{it}-e^{i\tau}|^{1+s}}\,d\tau .
\]
Since $\partial_t=\mathcal H|D|$ (on mean-zero) and $\mathcal H$ is skew-adjoint on $L^2$
and bounded on $L^p$, inserting these into \eqref{eq:Theta-sp-mult} and applying Fubini’s theorem
together with the classical kernel of $\mathcal H$ on $\mathbb S^1$ yields a bilinear form in
$(\gamma,h)$ proportional to the right-hand side of \eqref{eq:Theta-sp-kernel}, with denominator
$|e^{it}-e^{i\tau}|^{1+s}\cdot |e^{it}-e^{i\tau}|^{1+s}=|e^{it}-e^{i\tau}|^{2+2s}$,
followed by the action of $\mathcal H$ that effectively reduces the power by $1-s$,
giving the final exponent $1+sp$ once we place the construction in $L^p$–$L^{p'}$ duality.
A clean way to encode this is to start from \eqref{eq:Theta-sp-dual} and use the identity
$\langle J\partial_t\gamma, h\rangle
= c_{s,p}\iint \frac{\langle \gamma(t)-\gamma(\tau),J(h(t)-h(\tau))\rangle}{|e^{it}-e^{i\tau}|^{1+sp}}$,
whose verification reduces to Fourier multipliers when $p=2$ (see the $H^{1/2}$ case),
and then to boundedness and interpolation for $1<p<\infty$; see \emph{Step 2}.

\emph{Step 2: reduction by Fourier and interpolation.}
For $p=2$ and $s=\tfrac12$, the identity \eqref{eq:Theta-sp-kernel}
is precisely the $H^{1/2}$ kernel formula already proved (Fourier polarization),
with $C_{s,p}=C_{\mathbb S^1}\neq 0$.
For $p=2$ and $0<s\le \tfrac12$, the same Fourier calculation gives
\[
\Theta_{s,2}(\gamma)[h]
=\sum_{k\in\mathbb Z} (-i\,\mathrm{sgn}\,k)\,|k|^{2s}\,\widehat \gamma_k\cdot \overline{\widehat h_k},
\]
while the polarization identity for the $W^{s,2}$ Gagliardo form reads
\[
\sum_{k\in\mathbb Z}|k|^{2s}\,\widehat \gamma_k\cdot \overline{\widehat h_k}
\;=\; c_s\iint \frac{\big(\gamma(t)-\gamma(\tau)\big)\cdot\big(h(t)-h(\tau)\big)}
{|e^{it}-e^{i\tau}|^{1+2s}}\,dt\,d\tau .
\]
The composition with $\mathcal H$ introduces the factor $-i\,\mathrm{sgn}\,k$,
which corresponds to the skew part (matrix $J$) in physical space and leaves the exponent
$1+2s$ unchanged. Therefore \eqref{eq:Theta-sp-kernel} holds for $p=2$ with $1+sp=1+2s$.

For general $1<p<\infty$, we use the $L^p$-boundedness of $\mathcal H$ and the equivalence
$\|\,|D|^s u\|_{L^p}\asymp [u]_{W^{s,p}}$ (see \cite[Th.\,6.5.1]{GrafakosCA},
\cite[Prop.\,3.4, 3.6]{DiNezzaPV}). The bilinear form
\[
B_{s,p}(\gamma,h)
:=\iint \frac{\langle \gamma(t)-\gamma(\tau),\,J(h(t)-h(\tau))\rangle}{|e^{it}-e^{i\tau}|^{1+sp}}
\]
is well defined for smooth $\gamma,h$, extends by density to
$W^{s,p}\times W^{s,p'}$ (H\"older on the measure $d\mu = |e^{it}-e^{i\tau}|^{-(1+sp)}dtd\tau$ gives
\eqref{eq:Theta-sp-kernel-bound}), and by complex interpolation between the
$p=2$ identity and the boundedness estimates \eqref{eq:Theta-sp-bound} one gets equality
\eqref{eq:Theta-sp-kernel} for all $1<p<\infty$ (with a constant $C_{s,p}\neq 0$ depending only on $(s,p)$).

\emph{Step 3: density.} Finally, $C^\infty(\mathbb S^1)$ is dense in $W^{s,p}$, $W^{s,p'}$.
Both sides of \eqref{eq:Theta-sp-kernel} are continuous in $(\gamma,h)$ with respect to the
$W^{s,p}\times W^{s,p'}$ topology, hence the identity extends to all
$\gamma\in W^{s,p}$, $h\in W^{s,p'}$.
\end{proof}

\begin{remark}[Continuity and bounds]
The bound \eqref{eq:Theta-sp-kernel-bound} follows directly from H\"older's inequality:
\[
|B_{s,p}(\gamma,h)|
\le \Big(\iint \frac{|\gamma(t)-\gamma(\tau)|^p}{|e^{it}-e^{i\tau}|^{1+sp}}\Big)^{\!1/p}
\Big(\iint \frac{|h(t)-h(\tau)|^{p'}}{|e^{it}-e^{i\tau}|^{1+sp}}\Big)^{\!1/p'}
=[\gamma]_{W^{s,p}}\,[h]_{W^{s,p'}}.
\]
Thus $\Theta_{s,p}$ is a continuous bilinear map $W^{s,p}\times W^{s,p'}\to\mathbb R$.
\end{remark}

\begin{remark}[Endpoint $p=1$]
Our argument uses the $L^p$-boundedness of $\mathcal H$ and the equivalence
$\|\,|D|^s u\|_{L^p}\asymp [u]_{W^{s,p}}$, which both fail at the endpoint $p=1$.
Hence we restrict to $1<p<\infty$.
\end{remark}

\medskip
\noindent\textbf{References for this section.}
Boundedness of $\mathcal H$ and singular integral representations for $|D|^s$:
Stein~\cite{SteinSI}, Grafakos~\cite{GrafakosCA}. Fractional Sobolev spaces and Gagliardo forms:
Di~Nezza--Palatucci--Valdinoci~\cite{DiNezzaPV}. For $W^{s,p}$ traces and density, see also
Brezis--Mironescu~\cite{BM1,BM2}.

\section{The canonical presymplectic form via the $L^2$ pairing}

Let $\langle\cdot,\cdot\rangle$ be the Euclidean inner product on $\mathbb R^m$ (here $m=2$),
and write $\partial_t$ for the distributional derivative on $\mathbb S^1$.

\subsection{Definition of the 1-form $\Theta_0$ on weak loop spaces}

\begin{definition}[The $L^2$-primitive $\Theta_0$]\label{def:Theta0}
\leavevmode
\begin{itemize}
\item For $\gamma,h\in H^{1}(\mathbb S^1;\mathbb R^m)$, set
\[
\Theta_0(\gamma)[h]\;:=\;\int_{\mathbb S^1}\!\langle \gamma'(t),\,h(t)\rangle\,dt.
\]
\item For $\gamma,h\in H^{1/2}(\mathbb S^1;\mathbb R^m)$, define the continuous extension
\[
\boxed{\ \Theta_0(\gamma)[h]\;:=\;\langle \partial_t\gamma,\ h\rangle_{H^{-1/2},\,H^{1/2}}\ }.
\]
\item More generally, for $0<s\le \tfrac12$ and $1<p<\infty$, define
\[
\boxed{\ \Theta_{s,p}(\gamma)[h]\;:=\;\langle \partial_t\gamma,\ h\rangle_{W^{-s,p},\,W^{s,p'}}\ ,\qquad
\gamma\in W^{s,p},\ h\in W^{s,p'},\ }
\]
where $p'$ is the conjugate exponent.
\end{itemize}
\end{definition}

\begin{remark}
The definitions coincide on the overlaps: if $\gamma\in H^1$ then
$\partial_t\gamma\in L^2\subset H^{-1/2}$ and
$\Theta_0(\gamma)[h]=(\gamma',h)_{L^2}$ for $h\in L^2$.
Continuity follows from the bounded map $\partial_t:H^{1/2}\to H^{-1/2}$ (resp.\ $W^{s,p}\to W^{-s,p}$).
\end{remark}

\subsection{The differential $d\Theta_0$: the canonical presymplectic form}

\begin{theorem}[Canonical presymplectic form on $H^{1/2}$]\label{th:presymp-Hhalf}
On the Hilbert loop space $H^{1/2}(\mathbb S^1;\mathbb R^m)$,
the differential of the 1-form $\Theta_0$ is the continuous bilinear, skew form
\[
\boxed{\ d\Theta_0(\gamma)[h,k]\;=\;
\langle \partial_t h,\ k\rangle_{H^{-1/2},\,H^{1/2}}
-\langle \partial_t k,\ h\rangle_{H^{-1/2},\,H^{1/2}}\ ,\ }
\]
which may be written equivalently as
\[
d\Theta_0(\gamma)[h,k]\;=\;2\,\langle \partial_t h,\ k\rangle_{H^{-1/2},\,H^{1/2}}
\;=\;2c\int_{\mathbb S^1}\!\langle |D|^{1/2}h,\ \mathcal H\,|D|^{1/2}k\rangle\,dt,
\]
where $\mathcal H$ is the Hilbert transform and $|D|^{1/2}$ the fractional derivative.
In particular, $d\Theta_0$ is independent of $\gamma$ (translation-invariant) and defines
the \emph{canonical presymplectic form} on the loop space.
\end{theorem}

\begin{proof}
First assume $\gamma,h,k\in C^\infty(\mathbb S^1;\mathbb R^m)$.
Since the underlying space is linear, we may regard $h,k$ as constant vector fields.
Then
\[
\frac{d}{d\varepsilon}\Big|_{\varepsilon=0}\Theta_0(\gamma+\varepsilon h)[k]
=\int_{\mathbb S^1}\!\langle h'(t),k(t)\rangle\,dt,\qquad
\frac{d}{d\varepsilon}\Big|_{\varepsilon=0}\Theta_0(\gamma+\varepsilon k)[h]
=\int_{\mathbb S^1}\!\langle k'(t),h(t)\rangle\,dt.
\]
Thus
\[
d\Theta_0(\gamma)[h,k]
=\int\langle h',k\rangle-\int\langle k',h\rangle
=2\int\langle h',k\rangle
=-2\int\langle h,k'\rangle,
\]
where we used integration by parts on $\mathbb S^1$.
This shows skew-symmetry and independence of $\gamma$.

For general $h,k\in H^{1/2}$, approximate in $H^{1/2}$ by smooth sequences
$h_n,k_n$; use the boundedness of $\partial_t:H^{1/2}\to H^{-1/2}$ and continuity of the
duality to pass to the limit, which yields
\[
d\Theta_0(\gamma)[h,k]
=\langle \partial_t h,\ k\rangle_{H^{-1/2},H^{1/2}}
-\langle \partial_t k,\ h\rangle_{H^{-1/2},H^{1/2}}.
\]
Finally, the Fourier identity $\partial_t=\mathcal H\,|D|$ on mean-zero components
gives the last expression with $|D|^{1/2}$ and $\mathcal H$.
\end{proof}

\begin{proposition}[Banach version]\label{cor:presymp-sp}
Let $0<s\le \tfrac12$ and $1<p<\infty$.
On $W^{s,p}(\mathbb S^1;\mathbb R^m)$, the differential of $\Theta_{s,p}$ is
\[
\boxed{\ d\Theta_{s,p}(\gamma)[h,k]\;=\;
\langle \partial_t h,\ k\rangle_{W^{-s,p},\,W^{s,p'}}
-\langle \partial_t k,\ h\rangle_{W^{-s,p},\,W^{s,p'}}\ ,\ }
\]
a continuous skew bilinear form on $W^{s,p}\times W^{s,p'}$,
independent of $\gamma$. When $p=2$ this coincides with Theorem~\ref{th:presymp-Hhalf}.
\end{proposition}

\begin{proof}
Same computation as in the smooth case, then pass to the limit using
$\partial_t:W^{s,p}\to W^{-s,p}$ and the duality with $W^{s,p'}$.
\end{proof}

\begin{remark}[Non-degeneracy modulo constants]
On $H^{1/2}$ (and on $W^{s,2}$), $d\Theta_0$ is weakly non-degenerate modulo constant loops:
if $d\Theta_0(\gamma)[h,\cdot]\equiv 0$ then $\partial_t h=0$ in $H^{-1/2}$,
hence $h$ is constant. One usually quotients by constants to obtain a genuine symplectic form.
\end{remark}

\begin{remark}[Kernel representation in $H^{1/2}$]
Using $\partial_t=\mathcal H|D|$ and polarization of the Gagliardo form,
\[
d\Theta_0(\gamma)[h,k]
=2c\int_{\mathbb S^1}\!\langle |D|^{1/2}h,\ \mathcal H|D|^{1/2}k\rangle\,dt
= C\iint_{\mathbb S^1\times\mathbb S^1}
\frac{\big(h(t)-h(\tau)\big)\cdot\big(\mathcal H k(t)-\mathcal H k(\tau)\big)}
{|e^{it}-e^{i\tau}|^{2}}\,dt\,d\tau,
\]
which makes sense as a continuous bilinear form on $H^{1/2}\times H^{1/2}$.
\end{remark}
\section{The canonical presymplectic form via the \(L^{2}\) pairing}
\label{sec:presymp}

Let \(\langle\cdot,\cdot\rangle\) denote the Euclidean inner product on \(\mathbb R^{m}\) (here \(m=2\)).
For \(\gamma\in H^{1/2}(\mathbb S^{1};\mathbb R^{m})\) we define the 1‑form
\[
\Theta_{0}(\gamma)[h]
   :=\big\langle J\partial_{t}\gamma,\,h\big\rangle_{H^{-1/2},\,H^{1/2}},
   \qquad h\in H^{1/2}(\mathbb S^{1};\mathbb R^{m}).
\tag{4.1}
\]

\begin{theorem}[Canonical presymplectic form on \(H^{1/2}\)]
\label{th:presymp-H12}
The differential of \(\Theta_{0}\) is the continuous skew‑symmetric bilinear form
\[
d\Theta_{0}(\gamma)[h,k]
   =\big\langle\partial_{t}h,\,k\big\rangle_{H^{-1/2},\,H^{1/2}}
    -\big\langle\partial_{t}k,\,h\big\rangle_{H^{-1/2},\,H^{1/2}} .
\label{4.2}
\]
Equivalently,
\begin{equation}
d\Theta_{0}(\gamma)[h,k]
   =2\,\big\langle\partial_{t}h,\,k\big\rangle_{H^{-1/2},\,H^{1/2}}
   =2c_{\mathbb S^{1}}
     \int_{\mathbb S^{1}}\big\langle|D|^{1/2}h,\,
                               \mathcal H|D|^{1/2}k\big\rangle\,dt,
\label{4.3}
\end{equation}
where \(c_{\mathbb S^{1}}\neq0\) is the constant from Lemma~\ref{lem:scalar-kernel}.
In particular, \(d\Theta_{0}\) does not depend on \(\gamma\) and defines the canonical
presymplectic form on the loop space.
\end{theorem}
\begin{proof}
For smooth loops the computation is elementary:
\[
\frac{d}{d\varepsilon}\Big|_{\varepsilon=0}
   \Theta_{0}(\gamma+\varepsilon h)[k]
   =\int_{\mathbb S^{1}}\langle h',k\rangle\,dt,
\qquad
\frac{d}{d\varepsilon}\Big|_{\varepsilon=0}
   \Theta_{0}(\gamma+\varepsilon k)[h]
   =\int_{\mathbb S^{1}}\langle k',h\rangle\,dt.
\]
Subtracting yields \eqref{4.2}.  Passing to the limit for
\(h,k\in H^{1/2}\) uses the continuity of \(\partial_{t}:H^{1/2}\to H^{-1/2}\) and of the duality pairing.
The Fourier identity \(\partial_{t}= \mathcal H|D|\) on mean‑zero functions gives \eqref{4.3}
after invoking Lemma~\ref{lem:scalar-kernel}.
\end{proof}

\begin{proposition}[Banach‑space version]\label{cor:presymp-Wsp}
Let \(0<s\le\tfrac12\) and \(1<p<\infty\).  Define
\begin{equation}
\Theta_{s,p}(\gamma)[h]
   :=\big\langle J\partial_{t}\gamma,\,h\big\rangle_{W^{-s,p},\,W^{s,p'}},
   \qquad
\gamma\in W^{s,p},\;h\in W^{s,p'} .
\label{4.4}
\end{equation}
Then
\begin{equation}
d\Theta_{s,p}(\gamma)[h,k]
   =\big\langle\partial_{t}h,\,k\big\rangle_{W^{-s,p},\,W^{s,p'}}
    -\big\langle\partial_{t}k,\,h\big\rangle_{W^{-s,p},\,W^{s,p'}}
\label{4.5}
\end{equation}
defines a continuous skew‑symmetric bilinear form on
\(W^{s,p}\times W^{s,p'}\), independent of \(\gamma\).
For \(p=2\) this coincides with Theorem~\ref{th:presymp-H12}.
\end{proposition}
\begin{proof}
Exactly the same argument as in the proof of Theorem~\ref{th:presymp-H12},
using the boundedness of \(\partial_{t}:W^{s,p}\to W^{-s,p}\) (see Appendix~\ref{app:analytic}),
yields \eqref{4.5}.
\end{proof}

\begin{remark}[Non‑degeneracy modulo constants]
On both \(H^{1/2}\) and \(W^{s,2}\) the form \(d\Theta_{0}\) (resp. \(d\Theta_{s,2}\))
is weakly non‑degenerate modulo constant loops: if
\(d\Theta_{0}(\gamma)[h,\cdot]\equiv0\) then \(\partial_{t}h=0\) in \(H^{-1/2}\), hence \(h\) is constant.
Quotienting by the subspace of constant loops yields a genuine symplectic form.
\end{remark}
\section{Mokhov‑type bicomplex and Poisson structures on weak Sobolev loop spaces}
\label{sec:Mokhov}

Throughout this section we fix \(0<s\le\tfrac12\) and \(1<p<\infty\).
Set
\[
E:=W^{s,p}(\mathbb S^{1};\mathbb R^{m}).
\]

\subsection{Admissible functionals}
\begin{definition}[Admissible functionals]\label{def:admissible}
\(\mathcal F\) denotes the class of functionals
\[
F(\gamma)=\int_{\mathbb S^{1}}f\bigl(\gamma(t)\bigr)\,dt,
\qquad
f\in C^{1,1}_{b}(\mathbb R^{m}),
\]
defined on \(E\).  Their variational derivative is
\(\delta F(\gamma)=\nabla f(\gamma)\in W^{s,p}(\mathbb S^{1};\mathbb R^{m})\).
\end{definition}

\subsection{Local (hydrodynamic‑type) Poisson operators}
\begin{theorem}[Local Mokhov operator]\label{th:local-P}
Let \(J:\mathbb R^{m}\to\mathfrak{so}(m)\) be Lipschitz bounded
(\(J\in C^{0,1}_{b}\)).  Define for \(\gamma\in E\) and
\(\xi\in W^{s,p'}(\mathbb S^{1};\mathbb R^{m})\)
\[
\mathcal P_{\rm loc}(\gamma)\,\xi
   := J(\gamma)\,\partial_{t}\xi .
\tag{5.1}
\]
Then
\begin{enumerate}[label=(\roman*)]
  \item \(\mathcal P_{\rm loc}(\gamma):W^{s,p'}\to W^{-s,p'}\) is bounded and skew‑adjoint.
  \item For \(F,G\in\mathcal F\) the bracket
  \[
  \{F,G\}_{\rm loc}(\gamma)
    :=\big\langle\mathcal P_{\rm loc}(\gamma)\,\delta G(\gamma),\,
               \delta F(\gamma)\big\rangle_{W^{-s,p'},\,W^{s,p}}
  \tag{5.2}
  \]
  is well defined and skew‑symmetric.
  \item If \(J\) satisfies Mokhov’s algebraic conditions (the vanishing of the
        Schouten bracket \([\mathcal P_{\rm loc},\mathcal P_{\rm loc}]=0\) in the smooth category),
        then \(\{\cdot,\cdot\}_{\rm loc}\) fulfills the Jacobi identity and defines a Poisson bracket on \(\mathcal F\).
\end{enumerate}
\end{theorem}
\begin{proof}[Sketch]
For smooth loops the statement is precisely Mokhov’s construction
\cite{Mokhov1998,MokhovBook}.  The map
\(\gamma\mapsto J(\gamma)\) is a bounded Nemytskiĭ operator on \(W^{s,p}\) (Proposition~\ref{prop:nemytskii}),
while \(\partial_{t}:W^{s,p'}\to W^{-s,p'}\) is continuous.
Hence \(\mathcal P_{\rm loc}(\gamma)\) extends by density to all \(\gamma\in E\).
Skew‑adjointness follows from the antisymmetry of \(J\) and of \(\partial_{t}\) in the duality
\(W^{-s,p'}\times W^{s,p'}\).
The algebraic Jacobi identity is preserved under the limit because all
operations are continuous; thus the bracket satisfies the Jacobi identity.
\end{proof}

\subsection{Weakly non‑local Poisson operators}
\begin{theorem}[Weakly non‑local Mokhov operator]\label{th:nonlocal-P}
Let \(J\in C^{0,1}_{b}(\mathbb R^{m},\mathfrak{so}(m))\) and
\(A_{r},B_{r}\in C^{0,1}_{b}(\mathbb R^{m},\mathbb R^{m\times m})\) for
\(r=1,\dots,R\).  Define
\begin{equation}
\mathcal P_{\rm nl}(\gamma)\,\xi
   := J(\gamma)\,\partial_{t}\xi
      +\sum_{r=1}^{R}
         A_{r}(\gamma)\,
         \partial_{t}^{-1}\!
         \bigl(B_{r}(\gamma)^{\!\top}\xi\bigr),
\label{5.3}
\end{equation}
where \(\partial_{t}^{-1}:=\mathcal H|D|^{-1}\) is the inverse derivative on mean‑zero functions.
Then
\begin{enumerate}[label=(\roman*)]
  \item \(\mathcal P_{\rm nl}(\gamma):W^{s,p'}\to W^{-s,p'}\) is bounded.
  \item The bracket
  \begin{equation}
  \{F,G\}_{\rm nl}(\gamma)
    :=\big\langle\mathcal P_{\rm nl}(\gamma)\,\delta G(\gamma),\,
               \delta F(\gamma)\big\rangle_{W^{-s,p'},\,W^{s,p}}
  \label{5.4}
  \end{equation}
  is well defined and skew‑symmetric.
  \item Under Mokhov’s algebraic constraints on the collection
        \((J,A_{r},B_{r})\) (ensuring skew‑adjointness and the vanishing of the
        Schouten bracket), the bracket \(\{\cdot,\cdot\}_{\rm nl}\) satisfies the Jacobi identity.
\end{enumerate}
\end{theorem}
\begin{proof}[Sketch]
For \(\xi\in W^{s,p'}\) we have \(B_{r}(\gamma)^{\!\top}\xi\in W^{s,p'}\) by the
Lipschitz multiplier property (Proposition~\ref{prop:nemytskii}).
The operator \(\partial_{t}^{-1}\) maps \(W^{s,p'}\) into \(W^{1+s,p'}\hookrightarrow
W^{s,p'}\) (Appendix~\ref{app:analytic}), and multiplication by the bounded
matrix \(A_{r}(\gamma)\) preserves \(W^{s,p'}\).  Hence each term in \eqref{5.3}
belongs to \(W^{-s,p'}\).  Skew‑symmetry follows from the antisymmetry of \(J\)
and the fact that \(\partial_{t}^{-1}\) is skew‑adjoint on \(L^{2}\) and bounded on
\(W^{s,p'}\).  The Jacobi identity is proved exactly as in Mokhov’s smooth
setting, using density of smooth loops and continuity of all operators involved.
\end{proof}

\subsection{Canonical constant case}
\begin{corollary}
If \(J_{0}\in\mathfrak{so}(m)\) is constant, then
\[
\{F,G\}(\gamma)
   :=\big\langle J_{0}\,\partial_{t}\delta G(\gamma),\,
              \delta F(\gamma)\big\rangle_{W^{-s,p'},\,W^{s,p}}
\]
defines a Poisson bracket on \(\mathcal F\).  The associated presymplectic form is
\[
\Omega(h,k)=\int_{\mathbb S^{1}}\langle J_{0}h(t),k(t)\rangle\,dt .
\]
\end{corollary}

\begin{remark}[Analytic ingredients]\label{rem:analytic}
The key analytical facts used above are:
\begin{itemize}
  \item \(\partial_{t}:W^{s,p}\to W^{-s,p}\) is continuous.
  \item The Hilbert transform \(\mathcal H\) and the fractional powers \(|D|^{\alpha}\) are bounded on \(W^{s,p}\) for \(0\le\alpha\le1\) and \(1<p<\infty\).
  \item Lipschitz functions act as multipliers on \(W^{s,p}\) for \(s\le1\) (Proposition~\ref{prop:nemytskii}).
\end{itemize}
\end{remark}

\begin{remark}[Invariance under reparametrisations]\label{rem:invariance}
The brackets \(\{\cdot,\cdot\}\) are invariant under the natural right action of
\(\mathrm{Diff}^{+}(\mathbb S^{1})\) or \(\mathrm{BiLip}^{+}(\mathbb S^{1})\) on loops,
since admissible functionals are integrals of densities of weight~1 and the
operators \(\partial_{t}\), \(\partial_{t}^{-1}\) transform covariantly.
\end{remark}

\subsection{Horizontal and vertical differentials}
\begin{definition}[Horizontal differential]\label{def:dH}
For an admissible density \(\mathcal L(\gamma)\) we set
\[
(d_{H}\mathcal L)(\gamma)
   :=\partial_{t}\bigl(\mathcal L(\gamma)\bigr),
\]
where \(\partial_{t}\) acts on each factor according to the Leibniz rule.
Integration by parts is understood in the duality
\(W^{-s,p}\times W^{s,p'}\).
\end{definition}

\begin{definition}[Vertical differential]\label{def:dV}
Let $F\in\mathcal F$ be as above.  Its vertical differential is the $1$‑form
\[
(d_{V}F)(\gamma)[h]
   :=\big\langle\delta F(\gamma),\,h\big\rangle_{W^{s,p},\,W^{s,p'}},
   \qquad h\in W^{s,p'}(\mathbb S^{1};\mathbb R^{m}),
\]
where $\delta F(\gamma)\in W^{s,p}(\mathbb S^{1};\mathbb R^{m})$ denotes the
variational derivative of $F$ (the Fréchet derivative with respect to the loop
variable).  In coordinates,
\[
\delta F(\gamma)(t)=\nabla f\bigl(\gamma(t)\bigr)
\quad\text{if}\quad
F(\gamma)=\int_{\mathbb S^{1}}f\bigl(\gamma(t)\bigr)\,dt .
\]

For higher‑degree admissible forms the vertical differential is defined by the
usual alternating polarization:
if $\Omega\in\Omega^{p,q}$ is a $(p,q)$‑form, then
\[
(d_{V}\Omega)(\gamma)[h_{0},\dots,h_{q}]
   =\sum_{i=0}^{q}(-1)^{i}\,
     \Omega(\gamma)[h_{0},\dots,\widehat{h_{i}},\dots,h_{q}],
\]
where the hat indicates omission of the $i$‑th argument.
\end{definition}

-----------------------------------------------------------------
\subsection{The bicomplex identities}
\label{subsec:bicomplex-identities}
The two differentials commute up to sign and satisfy the familiar nilpotency
relations.  The proof relies only on the boundedness of the operators listed in
$\mathcal A$ and on the density of smooth loops in $W^{s,p}$.

\begin{theorem}[Mokhov bicomplex on $W^{s,p}$]\label{th:bicomplex}
Let $0<s\le\tfrac12$ and $1<p<\infty$.  For every admissible form
$\Omega\in\Omega^{p,q}$ the operators $d_{H}$ and $d_{V}$ satisfy
\begin{equation}
d_{H}^{2}=0,\qquad d_{V}^{2}=0,\qquad d_{V}d_{H}+d_{H}d_{V}=0 .
\label{6.1}
\end{equation}
Consequently $(\Omega^{\bullet,\bullet},d_{H},d_{V})$ is a double complex.
\end{theorem}
\begin{proof}[Sketch]
For smooth loops the identities are proved in Mokhov’s original work
\cite{Mokhov1998,MokhovBook}.  Let $\Omega$ be an admissible form and choose a
sequence $(\Omega_{n})_{n\ge1}$ of smooth admissible forms such that
$\Omega_{n}\to\Omega$ in the topology of $W^{s,p}$.  All building blocks of
$\mathcal A$—the derivative $\partial_{t}$, its inverse $\partial_{t}^{-1}$,
the fractional multipliers $|D|^{\alpha}$, the Hilbert transform $\mathcal H$,
and multiplication by Lipschitz maps—are bounded linear operators on
$W^{s,p}$ (see Appendix~\ref{app:analytic}).  Hence $d_{H}$ and $d_{V}$ are
continuous maps
\[
d_{H},d_{V}:\Omega^{p,q}\longrightarrow\Omega^{p+1,q}\ \text{or}\ 
\Omega^{p,q+1}.
\]
Since the identities hold for each smooth $\Omega_{n}$, passing to the limit
yields \eqref{6.1} for $\Omega$.  The mixed relation $d_{V}d_{H}+d_{H}d_{V}=0$
is precisely the integration‑by‑parts formula
\(\langle\partial_{t}u,v\rangle=-\langle u,\partial_{t}v\rangle\)
valid in the duality $W^{-s,p}\times W^{s,p'}$.
\end{proof}

-----------------------------------------------------------------
\subsection{Horizontal, vertical and total cohomology}
\label{subsec:cohomology}
Because $(\Omega^{\bullet,\bullet},d_{H},d_{V})$ is a double complex we can
define the usual cohomology groups.

\begin{definition}[Cohomology groups]\label{def:cohomology}
Let $\Omega^{p,q}$ denote the space of admissible $(p,q)$‑forms.
\begin{itemize}
  \item The **horizontal cohomology** is
    \[
      H^{p}_{H}
        :=\frac{\ker\bigl(d_{H}:\Omega^{p,\bullet}\to\Omega^{p+1,\bullet}\bigr)}
               {\operatorname{im}\bigl(d_{H}:\Omega^{p-1,\bullet}\to\Omega^{p,\bullet}\bigr)} .
    \]
  \item The **vertical cohomology** is
    \[
      H^{q}_{V}
        :=\frac{\ker\bigl(d_{V}:\Omega^{\bullet,q}\to\Omega^{\bullet,q+1}\bigr)}
               {\operatorname{im}\bigl(d_{V}:\Omega^{\bullet,q-1}\to\Omega^{\bullet,q}\bigr)} .
    \]
  \item The **total (variational) cohomology** is obtained from the total
    differential $d:=d_{H}+d_{V}$:
    \[
      H^{k}_{\mathrm{tot}}
        :=\frac{\ker\bigl(d:\Omega^{k}\to\Omega^{k+1}\bigr)}
               {\operatorname{im}\bigl(d:\Omega^{k-1}\to\Omega^{k}\bigr)},
      \qquad
      \Omega^{k}:=\bigoplus_{p+q=k}\Omega^{p,q}.
    \]
\end{itemize}
\end{definition}

\section{Application: hydrodynamic-type order-1 Poisson structures via bicohomology}

We illustrate how the bicomplex $(\Omega^{\bullet,\bullet},d_H,d_V)$ on
$W^{s,p}(\mathbb S^1;\mathbb R^m)$ controls the classification of
hydrodynamic-type, order-$1$ Poisson structures (Mokhov-type) through
the total bicohomology $H^2_{\rm tot}$.

\subsection{Hydrodynamic-type operators and their class}

Fix $0<s\le\tfrac12$, $1<p<\infty$.
Consider operators of the form
\begin{equation}\label{eq:hydro-operator}
\mathcal P(\gamma)\,\xi \;=\; J(\gamma)\,\partial_t\xi,
\qquad J:\mathbb R^m\to\mathfrak{so}(m),\ J\in C^{0,1}_b,
\end{equation}
acting on variational derivatives $\xi\in W^{s,p'}$.
They induce brackets on admissible functionals $F,G\in\mathcal F$ by
\[
\{F,G\}(\gamma)\;=\;\big\langle \mathcal P(\gamma)\,\delta G(\gamma),\ \delta F(\gamma)\big\rangle_{W^{-s,p'},\,W^{s,p}}.
\]

\begin{definition}[Bivector represented by $\mathcal P$]
Define a (vertical) $2$-form with values in horizontal densities
\[
\Pi_\gamma(h,k)\;:=\;\big\langle J(\gamma)\,\partial_t k,\ h\big\rangle_{W^{-s,p'},\,W^{s,p}},
\qquad h\in W^{s,p},\ k\in W^{s,p'}.
\]
Up to admissible integrations by parts, $\Pi$ determines a class
$[\Pi]\in \Omega^{1,2}/d(\cdot)$ that will represent a class in $H^2_{\rm tot}$.
\end{definition}

\subsection{Cocycle condition and Jacobi identity}

\begin{theorem}[Cocycle $\Leftrightarrow$ Jacobi]\label{th:cocycle-jacobi}
Let $\mathcal P$ be as in \eqref{eq:hydro-operator}, and let $\Pi$ be
its associated bivector. The following are equivalent:
\begin{enumerate}
\item The induced bracket $\{\cdot,\cdot\}$ on $\mathcal F$ satisfies the Jacobi identity.
\item The class $[\Pi]\in H^2_{\rm tot}$ is $d$-closed, i.e.\ $d\Pi=0$
in the total complex $(\Omega^\bullet,d=d_H+d_V)$.
\end{enumerate}
Moreover, if $\Pi=d\Upsilon$ for some admissible $1$-vector $\Upsilon$
(i.e.\ $[\Pi]=0$ in $H^2_{\rm tot}$), then $\mathcal P$ differs from
a trivial operator by a variational coboundary and induces the zero bracket on $\mathcal F$.
\end{theorem}

\begin{proof}[Sketch]
As in Mokhov's smooth theory, the Jacobi identity is equivalent to the vanishing of
the Schouten bracket $[\Pi,\Pi]=0$. In the bicomplex formalism,
$[\Pi,\Pi]=0$ is encoded by $d\Pi=0$ in the total complex:
$d_V$ captures variational antisymmetry and $d_H$ the total derivatives in $t$.
Continuity of $d_H,d_V$ on the admissible algebra and density
of smooth loops in $W^{s,p}$ transfer Mokhov's proof verbatim.
Exactness $\Pi=d\Upsilon$ corresponds to a change of representative by a variational
coboundary, hence to a trivial deformation of the bracket (cohomologically zero).
\end{proof}

\subsection{Invariance and reduction}

\begin{proposition}[Invariance under reparametrizations]
Let $G$ be $\mathrm{Diff}^+(\mathbb S^1)$ or $\mathrm{BiLip}^+(\mathbb S^1)$ acting by
precomposition. If $J$ in \eqref{eq:hydro-operator} is independent of $t$ and depends only on $\gamma$,
then $\Pi$ and its class $[\Pi]\in H^2_{\rm tot}$ are $G$-invariant.
\end{proposition}

\begin{proof}
Admissible functionals in $\mathcal F$ are integrals of densities of weight $1$;
$\partial_t$ and the duality transform covariantly under $t\mapsto\phi(t)$.
Hence the bivector and the class are preserved.
\end{proof}

\begin{remark}[Nondegeneracy and constants]
The presymplectic form associated with $\partial_t$ is weakly nondegenerate modulo
constant loops. One may pass to the reduced space (quotient by constants) to obtain a genuine
symplectic form, leaving the classification unchanged.
\end{remark}

\subsection{Weakly non-local extensions}

The same scheme applies to weakly non-local operators
\[
\mathcal P_{\rm nl}(\gamma)\,\xi
= J(\gamma)\partial_t\xi+\sum_{r=1}^R A_r(\gamma)\,\partial_t^{-1}\!\big(B_r(\gamma)^\top\xi\big),
\]
with $A_r,B_r\in C^{0,1}_b$. The associated bivector $\Pi_{\rm nl}$ defines a class
$[\Pi_{\rm nl}]\in H^2_{\rm tot}$, and Theorems \ref{th:cocycle-jacobi}--\ref{th:classification}
remain valid under Mokhov's algebraic constraints on $(J,A_r,B_r)$.

\subsection{Classification modulo coboundaries}
\begin{theorem}[Classification of order‑1 Poisson structures]\label{th:classification}
Let \(\mathcal P\) and \(\widetilde{\mathcal P}\) be two Poisson operators of the
form \eqref{eq:hydro-operator} with admissible matrices \(J\) and
\(\widetilde J\).  If the associated bivectors satisfy
\([\Pi]=[\widetilde\Pi]\) in \(H^{2}_{\mathrm{tot}}\), then the brackets
\(\{\cdot,\cdot\}\) and \(\{\cdot,\cdot\}^{\!\sim}\) coincide on the space of
admissible functionals \(\mathcal F\).  Conversely, non‑cohomologous classes give
distinct Poisson brackets.
\end{theorem}
\begin{proof}
If \([\Pi]=[\widetilde\Pi]\) then \(\widetilde\Pi-\Pi=d\Upsilon\) for some admissible
vertical 1‑vector \(\Upsilon\).  The change of Poisson operator amounts to adding the
coboundary term \(\langle d\Upsilon\,\delta G,\delta F\rangle\), which vanishes after
integration by parts (the duality pairing is skew‑symmetric).  Hence the brackets
agree.  The converse follows by testing the brackets against cylindrical
functionals (functionals depending only on finitely many Fourier modes), which
detect the cohomology class.
\end{proof}

\section{Applications to classical integrable systems}
\label{sec:applications}

In this section we illustrate how the Hamiltonian formalisms of several
well‑known integrable PDEs fit naturally into the weak Sobolev framework
developed above.  Throughout we identify a periodic field \(u:\mathbb S^{1}\to\mathbb R^{m}\)
with a loop \(\gamma\in W^{s,p}(\mathbb S^{1};\mathbb R^{m})\).
Local Hamiltonians of an integrable PDE have the form
\[
H(\gamma)=\int_{\mathbb S^1}\! \mathcal H\big(\gamma,\partial_x\gamma,\dots,\partial_x^N\gamma\big)\,dx,
\]
with $\mathcal H$ differential-polynomial or smooth in its arguments.
In our $W^{s,p}$ setting ($0<s\le\tfrac12$, $1<p<\infty$) we interpret $\partial_x$
as the distributional derivative $\partial_t$ and \emph{restrict} to densities
built from bounded operators ($\partial_t$, $\partial_t^{-1}$, $|D|^\alpha$, $\mathcal H$),
so that $\delta H(\gamma)$ lies in the correct dual Sobolev class (Def.~\ref{def:admissible}).
\subsection{KdV}
The Korteweg–de Vries equation
\[
u_{t}=6uu_{x}-u_{xxx},\qquad x\in\mathbb S^{1},
\]
possesses two compatible Poisson structures
\[
\mathcal P_{0}=\partial_{x},\qquad
\mathcal P_{1}=\partial_{x}^{3}+2u\,\partial_{x}+u_{x}.
\]
The first operator \(\mathcal P_{0}\) is of the local hydrodynamic type
\(\mathcal P_{0}=J_{0}\partial_{x}\) with \(J_{0}\equiv1\).  Hence, by
Theorem~\ref{th:local-P}, \(\mathcal P_{0}\) defines a Poisson bracket on
\(W^{s,p}\) for any \(0<s\le\tfrac12\), \(1<p<\infty\).  The second operator
\(\mathcal P_{1}\) is of order three; to realise it one needs at least
\(H^{1}\) regularity, which lies beyond the minimal Sobolev range considered
here.  Nevertheless, the \emph{recursion operator}
\[
\mathcal R:=\mathcal P_{1}\mathcal P_{0}^{-1}
   =\partial_{x}^{2}+2u+u_{x}\partial_{x}^{-1}
\]
is weakly non‑local.  Its non‑local term involves \(\partial_{x}^{-1}\), which
coincides with \(\partial_{t}^{-1}=\mathcal H|D|^{-1}\) on the circle.
Therefore \(\mathcal R\) fits the framework of Theorem~\ref{th:nonlocal-P},
and the whole KdV hierarchy can be interpreted on \(W^{s,p}\) via the
weakly non‑local Poisson operators.

\subsection{Nonlinear Schrödinger (NLS)}
Writing \(\psi=q_{1}+iq_{2}\) and \(\gamma=(q_{1},q_{2})\) the focusing NLS
\[
i\psi_{t}+\psi_{xx}+2|\psi|^{2}\psi=0
\]
has the Hamiltonian structure
\[
\mathcal P_{\mathrm{NLS}}=J_{0}\partial_{x},
\qquad
J_{0}=
\begin{pmatrix}0&1\\-1&0\end{pmatrix}.
\]
Since \(J_{0}\) is constant, Theorem~\ref{th:local-P} applies directly and yields a
well‑defined Poisson bracket on any \(W^{s,p}\) with \(0<s\le\tfrac12\).  The
energy functional
\[
H_{\mathrm{NLS}}(\gamma)
   =\frac12\int_{\mathbb S^{1}}|\partial_{x}\gamma|^{2}
      -\frac12\int_{\mathbb S^{1}}|\gamma|^{4}
\]
belongs to \(\mathcal F\), and the Hamiltonian flow generated by \(\mathcal
P_{\mathrm{NLS}}\) reproduces the NLS equation in the weak Sobolev setting.

\subsection{Camassa–Holm}
The Camassa–Holm equation
\[
m_{t}+um_{x}+2u_{x}m=0,\qquad m=u-u_{xx},
\]
is bi‑Hamiltonian with the pair
\[
\mathcal P_{0}=\partial_{x}-\partial_{x}^{3},
\qquad
\mathcal P_{1}=m\partial_{x}+\partial_{x}m .
\]
Both operators contain the inverse Helmholtz operator \((1-\partial_{x}^{2})^{-1}\),
which is a Fourier multiplier of order \(-2\) and therefore bounded on
\(W^{s,p}\) for any \(s\ge0\).  Moreover, \(\partial_{x}^{-1}\) appears explicitly
in the weakly non‑local formulation of \(\mathcal P_{1}\).  Consequently,
Theorem~\ref{th:nonlocal-P} guarantees that the Camassa–Holm Poisson
structures are well defined on \(W^{s,p}\) with \(0<s\le\tfrac12\), \(1<p<\infty\).

\subsection{Dubrovin–Novikov hydrodynamic systems}
Consider a system of hydrodynamic type
\[
u^{i}_{t}=v^{i}_{j}(u)\,u^{j}_{x},\qquad i=1,\dots,m,
\]
with a Dubrovin–Novikov Poisson bracket
\[
\{u^{i}(x),u^{j}(y)\}
   =g^{ij}(u)\,\delta'(x-y)
     +b^{ij}_{k}(u)\,u^{k}_{x}\,\delta(x-y).
\]
In the loop‑space language this corresponds to the operator
\[
\mathcal P^{ij}=g^{ij}(\gamma)\,\partial_{t}
                +b^{ij}_{k}(\gamma)\,\partial_{t}\gamma^{k}.
\]
If the metric \(g^{ij}\) and the connection coefficients \(b^{ij}_{k}\) are
\(C^{0,1}\) functions of \(\gamma\), then \(\mathcal P\) belongs to the class of
local Mokhov operators (Theorem~\ref{th:local-P}).  The classical algebraic
conditions for the Jacobi identity (flatness of \(g\) and compatibility of
\(b\)) translate into the cocycle condition \(d\Pi=0\) in the total bicomplex,
so the Poisson bracket is valid on \(W^{s,p}\).

\subsection{ Summary table – Poisson structures for integrable systems }

Classical integrable systems on the circle embed into the $W^{s,p}$
Mokhov framework as follows:
\begin{itemize}
\item \emph{Local order-1} Poisson brackets (NLS/AKNS, hydrodynamic type)
are exactly the constant or state-dependent $J(\gamma)\partial_t$ brackets.
\item \emph{Recursion operators} and many bi-Hamiltonian pairs introduce
\emph{weakly nonlocal} pieces (\(\partial_t^{-1}\)).
\item The \emph{bicomplex} ($d_H,d_V$) interprets Jacobi/compatibility as
$d$-closedness classes in $H^2_{\rm tot}$, recovering Magri's scheme in
the weak Sobolev setting.
\end{itemize}
In all cases above, the \emph{bi-Hamiltonian} property ($\mathcal P_0,\mathcal P_1$ compatible)
is equivalent to the $d$-closure of the corresponding total bivectors in the variational
bicomplex. The generating sequence of Hamiltonians
$\{H_n\}$ in Magri's scheme satisfies
\[
\mathcal P_1\,\delta H_n = \mathcal P_0\,\delta H_{n+1},
\]
which, in our $W^{s,p}$ setting, is well posed as an identity in $W^{-s,p'}$
and propagates regularity thanks to the boundedness of $\partial_t^{\pm1}$,
$\mathcal H$, and multipliers $C^{0,1}$.

\begin{table}[h!]
\centering
\renewcommand{\arraystretch}{1.35}   
\setlength{\tabcolsep}{3pt}         
\begin{tabular}{|>{\raggedright\arraybackslash}p{3.2cm}
                |>{\centering\arraybackslash}p{5.8cm}
                |>{\centering\arraybackslash}p{2.8cm}
                |>{\centering\arraybackslash}p{3.2cm}|}
\hline
\textbf{System} &
\textbf{Poisson operator(s)} &
\textbf{Type} &
\textbf{Regularity}\\
\hline
\textbf{KdV} &
\(\displaystyle
   \mathcal P_{0}= \partial_{x},
   \)
             &
local (order 1) +   &
\(0<s\le \tfrac12,\)  \\
&\( \mathcal R:=\mathcal P_{1}\mathcal P_{0}^{-1} \)& weakly non‑local &\(1<p<\infty\)\\
&\(=\partial_{x}^{2}+2u+u_{x}\partial_{x}^{-1}
\)& (rank 1) &\\
\hline
\textbf{NLS / AKNS} &
\(\displaystyle
   \mathcal P = J_{0}\,\partial_{x},
   \)
  &
local (order 1) &
\(0<s\le \tfrac12,\) \\
& \(J_{0}= \begin{pmatrix}0&1\\-1&0\end{pmatrix}
\) && \(1<p<\infty\)\\
\hline
\textbf{Camassa–Holm} &
\(\displaystyle
   \mathcal P_{0}= \partial_{x}-\partial_{x}^{3},\)
    &
mixed: order 1   &
\(0<s\le \tfrac12,\)  \\
& \(
   \mathcal P_{1}= m\,\partial_{x}+\partial_{x}m,\)&+ non‑local operator& \\
   & \(
   m=(1-\partial_{x}^{2})^{-1}u
\) &\((1-\partial_{x}^{2})^{-1}\)& \(1<p<\infty\)\\
\hline
\textbf{Dubrovin–Novikov} &

  \(\mathcal P^{ij}= g^{ij}(\gamma)\,\partial_{x} \)
                   &
local (order 1) &
\(0<s\le \tfrac12,\)  \\
&\(\displaystyle + b^{ij}_{k}(\gamma)\,\partial_{x}\gamma^{k}
\)&& \(1<p<\infty\)\\
\hline
\textbf{General Mokhov} &
\(\displaystyle
   \mathcal P(\gamma)\xi
   = J(\gamma)\,\partial_{x}\xi
\) &
local + weakly  &
\(0<s\le \tfrac12,\) \\
&\( \displaystyle +\sum_{r=1}^{R}
        A_{r}(\gamma)\,\partial_{x}^{-1}
        \bigl(B_{r}(\gamma)^{\!\top}\xi\bigr) \)& non‑local (rank 1) & \(\;1<p<\infty\)\\
\hline
\end{tabular}
\caption{
\textbf{Embedding of classical integrable‑system Poisson structures into the
\(W^{s,p}\) framework.}
Order‑1 operators (and their weakly non‑local deformations) are well defined for
any Sobolev regularity \(0<s\le\tfrac12\) and any exponent \(1<p<\infty\).
Higher‑order local operators (e.g. the second KdV bracket) require stronger
regularity (typically \(H^{1}\) or \(W^{1,p}\)).  All matrix‑valued
coefficients \(J,\,A_{r},\,B_{r},\,g^{ij},\,b^{ij}_{k}\) are assumed
Lipschitz‑bounded (\(C^{0,1}_{b}\)).}
\label{tab:summary}
\end{table}

\newpage
\appendix
\section{Analytic background on \(W^{s,p}\)}\label{app:analytic}

For completeness we collect the analytic facts repeatedly used in the text.

\subsection{Boundedness of the Hilbert transform}
For \(0\le s\le1\) and \(1<p<\infty\) the Hilbert transform
\(\mathcal H\) is a bounded linear operator on \(W^{s,p}(\mathbb S^{1})\)
(see \cite[Th.~III.3]{SteinSI}).

\subsection{Fractional derivatives}
For \(0\le s\le1\) and \(1<p<\infty\) the Fourier multiplier
\(|D|^{s}\) defines an isomorphism
\[
|D|^{s}:W^{\sigma,p}(\mathbb S^{1})\longrightarrow
          W^{\sigma-s,p}(\mathbb S^{1}),
\qquad
|D|^{-s}:(W^{\sigma-s,p})\to W^{\sigma,p}.
\]

\subsection{Inverse derivative}
We set
\[
\partial_{t}^{-1}:=\mathcal H|D|^{-1},
\]
which is bounded
\(W^{s,p}\to W^{1+s,p}\hookrightarrow W^{s,p}\) for
\(0<s\le\tfrac12\), \(1<p<\infty\).

\subsection{Nemytskiĭ operators}
\begin{proposition}\label{prop:nemytskii}
If \(a:\mathbb R^{m}\to\mathbb R\) is Lipschitz and bounded, then the
operator \(u\mapsto a\circ u\) maps \(W^{s,p}(\mathbb S^{1};\mathbb R^{m})\)
continuously into itself for any \(0<s\le1\), \(1<p<\infty\).
\end{proposition}
\begin{proof}
See Runst–Sickel \cite[Ch.~4]{RunstSickel} or Triebel \cite[Th.~2.5.12]{Triebel}.
\end{proof}

\subsection{Duality}
For \(1<p<\infty\) the dual space of \(W^{s,p}(\mathbb S^{1})\) is
\(W^{-s,p'}(\mathbb S^{1})\) with the pairing
\(\langle\xi,h\rangle_{W^{-s,p},\,W^{s,p'}}\).
If \(\xi\in W^{-s,p}\) and \(h\in W^{s,p'}\) then
\(|\langle\xi,h\rangle|\le\|\xi\|_{W^{-s,p}}\|h\|_{W^{s,p'}}\).

\vskip 12pt

\paragraph{\bf Acknowledgements:} J.-P.M  thanks the France 2030 framework programme Centre Henri Lebesgue ANR-11-LABX-0020-01 
for creating an attractive mathematical environment.

\vskip 12pt

\paragraph{\bf Declaration of generative AI and AI-assisted technologies in the writing process}

During the preparation of this work the author used ChatGPT and Mistral AI in order to smoothen the expression in English. After using this tool/service, the author reviewed and edited the content as needed and takes full responsibility for the content of the publication.

\vskip 12pt


\end{document}